\documentclass{comjnlarxiv}
\usepackage{url}
\usepackage{algorithm}
\usepackage{subfloat,subfig}
\usepackage{algpseudocode}
\usepackage{booktabs}
\usepackage{tikz}
\algnewcommand{\IIf}[1]{\State\algorithmicif\ #1\ \algorithmicthen}
\algnewcommand{\EndIIf}{\unskip\ \algorithmicend\ \algorithmicif}
\usepackage{amsmath}
\DeclareMathOperator{\rem}{\mathbin{\%}}

\DeclareMathOperator{\distance}{\mathrm{distance}}
\usepackage{siunitx}
\usepackage{float}
\DeclareFixedFont{\ttb}{T1}{txtt}{bx}{n}{12} 
\DeclareFixedFont{\ttm}{T1}{txtt}{m}{n}{12}  

\usepackage{color}
\definecolor{deepblue}{rgb}{0,0,0.5}
\definecolor{deepred}{rgb}{0.6,0,0}
\definecolor{deepgreen}{rgb}{0,0.5,0}

\usepackage{listings}

\newcommand\pythonstyle{\lstset{
language=Python,
basicstyle=\small,
}}

\lstnewenvironment{python}[1][]
{
\pythonstyle
\lstset{#1}
}
{}

\newcommand\pythoninline[1]{{\pythonstyle\lstinline!#1!}}

\begin{document}

\title[Exact Short Products]{Exact Short Products From Truncated Multipliers}
\author{Daniel Lemire}
\affiliation{Data Science Research Center, Universit\'e du Qu\'ebec (TELUQ), Montreal, Quebec, H2S 3L5, Canada} \email{daniel@lemire.me}

\shortauthors{D. Lemire}
 


\begin{abstract}
We sometimes need to compute the most significant digits of
the product of small integers with a multiplier requiring much storage:
e.g., a large integer (e.g., $5^{100}$) or an irrational number ($\pi$). 
We only need to access the most significant digits of the multiplier---as long as the integers are sufficiently small.
We provide an efficient algorithm to compute  the range of integers given a truncated multiplier and a desired number of digits.

\end{abstract}

\maketitle
\section{Introduction}

In applications such as cryptography~\cite{hars2006applications}, digital signal processing~\cite{schulte1993truncated}, number serialization~\cite{10.1145/3192366.3192369,10.1145/3360595} or number parsing~\cite{lemire2021number}, we need
to efficiently compute the most significant digits or bits (binary digits)  of a product. We call such a partial product a \emph{short} product.

We consider the computation of the product between  small integers and a multiplier with many digits.
To illustrate the problem, suppose that we want to compute
the product of $\pi$ with an integer and get 10-digit accurate
answer. With the 13~most significant digits of $\pi$ ($3.141592653589$), 
we get an accurate product for all integers in $[1, 1198)$ no matter the missing digits of $\pi$.
See Table~\ref{table:acc}.
\begin{table}
\caption{Range of integers $w$ such that the 10~most significant digits
 of $w\times \pi$ are exact as a function of the number of most significant digits of $\pi$ used. 
 The bound is blind to the missing digits of $\pi$. \label{table:acc}}
\begin{tabular}{cl}
\toprule
digits of $\pi$ & interval for 10-digit accuracy \\\midrule
10 & $[1, 2)$ \\
11 & $[1, 14)$ \\
12 & $[1, 209)$ \\
13 & $[1, 1198)$ \\
14 & $[1, 18149)$ \\
15 & $[1, 26255)$ \\
16 & $[1, 1454833)$ \\
17 & $[1, 14920539)$ \\
18 & $[1, 14920539)$ \\
19 & $[1, 1963319607)$ \\
20 & $[1, 17329613732)$ \\\bottomrule
\end{tabular}
\end{table}


We are given only the most significant digits of the multiplier (i.e., a \emph{short multiplier}) and we want to know whether we can compute the most significant digits of the product  exactly, for some range of values of $w$. 
We provide an efficient logarithmic-time  algorithm to compute the exact \emph{range of validity} (\S~\ref{sec:validrange}). We make available our algorithms as part of an open-source Python library.\footnote{\url{https://github.com/lemire/exactshortlib}} 

\section{Related Work}

Given only a short multiplier of $z$, we seek the exact range
of values $w$ such that we can compute a given number of most significant digits of the product $w\times z$ using the short multiplier.
One potential application of our work is in number parsing: given the string \texttt{3e100}, we may want to convert it to a 64-bit binary floating-point number:  $3\times 10^{100}  \approx 7721336384202043\times 2^{281}$ where 7721336384202043 is the 53-bit significand chosen to best approximate $3\times 10^{100}$. We have that $3\times 10^{100} = 3 \times 5^{100} \times 2^{100} $. Ignoring the power of two, the significand of the binary floating-point number ($7721336384202043$) may be computed by multiplying the decimal significand (3) by the power of five ($5^{100}$) and selecting the 53~most significant bits. For speed, we may want to avoid computing $5^{100}$ and the full product $3 \times 5^{100}$. Instead, we want to just store the most significant bits of $5^{100}$~\cite{lemire2021number}. 
However, we also need to check the validity of the short multipliers to ensure that the computation of the most significant digits is exact. 

To our knowledge, our problem, the exact computation of the range of validity of a short multiplier, is novel.  Adams~\cite{10.1145/3192366.3192369,10.1145/3360595} considered a related problem in the context of number serialization. They bound the maximum and minimum of $ax \rem b$ over an interval starting at zero. With these bounds, they show that powers of five truncated to 128~bits are sufficient to convert 64-bit binary floating-point numbers into equivalent decimal numbers. 
Specifically, Lemma~3.6 from Adams~\cite{10.1145/3192366.3192369} \emph{computes a conservative approximation of the true minimum and maximum} of $ax \rem b$ for $x\in[0, M]$. In contrast,  we present a logarithmic-time algorithm (in \S~\ref{sec:enumerating}) that provides all of the minima and maxima. This exact result allows us to compute an exact range of validity for a short multiplier.

\section{Mathematical Preliminaries}
\label{sec:math}
We present our core mathematical notation and we review some elementary results. See Table~\ref{tab:conventions}. For simplicity, we avoid references to equivalence classes or other extraneous concepts.

Let $\lfloor x \rfloor$ be the largest integer smaller or equal to $x$. For a number $a$ and another number $b \neq 0$, we define the integer division $a \div b \equiv \lfloor a / b \rfloor$ and the  remainder $ a \rem  b  \equiv a - (a \div b ) \times b$.

We say that $b\neq 0$ divides $a$ if $a\rem b = 0$.
We write the greatest common divisor of integers $a$ and $b$ as $\gcd(a,b)$. We say that $a$ and $b$ are coprime if 
$\gcd(a,b)=1$. 

The smallest integer $a'$ such that 
$a \div b = a' \div b$ is $a'=(a \div b ) \times b= a - (a\rem b)$. 
We have that $(a+b)\rem M=(a\rem M+b\rem M)\rem M$. 

Consider a positive integer divisor $M$. When $a\rem M = 0$ then $(-a)\rem M = 0$. Otherwise, when $a\rem M \neq 0$, then we have $(-a)\rem M = M-(a\rem M)$.

The distance between two integers $a,b$ is often defined as
the absolute value of their difference $\vert a- b \vert$. We use a generalized measure: \begin{eqnarray*}\distance_M (a,b) \equiv \min((a-b)\rem M,(b-a)\rem M).\end{eqnarray*}

Because $(a+b)\rem M=(a\rem M+b\rem M)\rem M$, we have the following elementary results:
\begin{itemize}
    \item $a\rem M + b \rem M < M$ if and only if $a\rem M + b \rem M = (a + b) \rem M$,
    \item $a\rem M + b \rem M \geq  M$ if and only if  $a\rem M + b \rem M  = M + (a + b) \rem M$.
\end{itemize}
Similarly, we have:
\begin{itemize}
    \item $a\rem M+b\rem M< M$ if and only if $(a+b)\rem M\geq \max(a\rem M,b\rem M)$,
    \item $a\rem M+b\rem M\geq M$ if and only if $(a+b)\rem M< \min(a\rem M,b\rem M)$.
\end{itemize}
Further if $a\rem M>0$ and $b\rem M >0$ 
then  $a\rem M+b\rem M< M$ if and only if  $(a+b)\rem M> \max(a\rem M,b\rem M)$.
Given two integers $a, b$ and an integer divisor $M$, we either have
$(a-b) \rem M = a\rem M - b\rem M $ when $a\rem M \geq b\rem M$, 
or $(a-b )\rem M = M+a\rem M - b\rem M$ otherwise.

\begin{table}
    \centering
    \begin{tabular}{cl}\toprule
symbol  & meaning\\\midrule
$\lfloor x \rfloor$ &  largest integer no larger than $x$ \\
$\gcd(a,b)$ &  greatest common divisor of $a$ and $b$ \\
$a \div b$  & integer division of $a$ by $b$: $\lfloor a/b\rfloor$\\
$a \rem b$  & remainder of the division of $a$ by $b$\\
$z$  & a multiplier (e.g., large integer) \\
$w$  & integer to be multiplied by $z$ \\
$\alpha$  & integer corresponding to a minimum  \\
$\beta$  & integer corresponding to a maximum  \\
$M$  & integer divisor  \\
$\distance_M (a,b)$ & $ \min((a-b)\rem M,(b-a)\rem M)$\\
\bottomrule
    \end{tabular}
    \caption{Notational conventions.}
    \label{tab:conventions}
\end{table}

\section{Plan}

Our derivation is organized as follow.
\begin{enumerate}
\item 
We formalize the 
concept of most significant digits of a product in \S~\ref{sec:mostsignificant}. Unsurprisingly, the computation of the most significant digits depends on the size of the product. Effectively, we get the
most significant digits by dividing by some power of the base (e.g., $10 ^3$), and by discarding the remainder.
\item In \S~\ref{sec:shortmultipliers}, we show that
a short multiplier always provides an exact short product if and only if the discarded remainder is not too large (Lemma~\ref{lemma:exact}). 
These remainders take the form  $(w\times z)\div M$  where
$w$ is an integer variable while $z$ and $M$ are integer constants.
By combining the result from the previous section (\S~\ref{sec:mostsignificant}), we describe how checking for exact most-significant digits requires  to bound various remainders over ranges.
We conclude this section with a technical lemma (Lemma~\ref{lemma:helpful}) which suggests that identifying the maxima of remainders is sufficient. Treating  $w\to (w\times z)\div M$ as a function, we need to identify the values of $w$ that provide a new maximum for the expression $(w\times z)\div M$ as $w$ is incremented (e.g., $w=1,2,3,\ldots$).
\item In \S~\ref{sec:enumerating}, we present Lemma~\ref{lemma:technical} which  says that if we know the location of the last maximum ($w=\beta$) and the last minimum ($w=\alpha$), then the next extremum is at the sum of the two ($w=\alpha + \beta$). This leads us to an efficient (logarithmic-time) algorithm to locate all extrema. The result is an algorithm that we present in Python (Fig.~\ref{fig:gaps}): it computes 
the \emph{gaps} between the extrema of $(w \times z)\rem M$ for $w=1,2,\ldots$ In \S~\ref{sec:offset}, we show that these gaps are sufficient to locate efficiently the extrema of  $(w \times z )\rem M$ over a range ($w=A, A+1, \ldots, B$) that does not begin at 1 ($w=1$). Finally, we provide a function \texttt{find\_max\_min} (Fig.~\ref{fig:extrema}) to 
enumerate all extrema  of $(w\times z) \rem M$ for $w = A,\ldots, B$.
\item 
In \S~\ref{sec:computingtherange}, we combine the function \texttt{find\_max\_min} with the results from \S~\ref{sec:shortmultipliers} to arrive at our main function (\texttt{find\_range}) presented in Fig.~\ref{fig:extrema}. Given a short multiplier, and a desired number of digits, it computes the range of validity.
\end{enumerate}
\section{Most Significant Digits}
\label{sec:mostsignificant}
We often represent integers with digits. E.g., the integer 1234 has four~decimal digits. The integer~7 requires three binary digits. The number of digits of the integer $x$ in base $b$ is the smallest integer $d$ such that $x \div b^d = 0$. By convention, the integer 0 has no digit (zero digit) and we do not consider negative integers. We may compute  the number of digits in base $b$ of a positive integer $x$ using the formula $\lceil \log_b (x+1) \rceil$. In base 10, the integers with three digits go from 100 to 999, or from $10^2$ to $10^{3}-1$, inclusively. More generally, an integer has $d$~digits in base $b$ if it is between $b^{d-1}$ and $b^d-1$, inclusively.

The product between an integer having $d_1$~digits and integer having $d_2$~digits is between $b^{d_1+d_2-2}$ and $b^{d_1+d_2}-b^{d_1}-b^{d_2}+1$. (inclusively). Thus the product has either $d_1+d_2-1$~digits or $d_1+d_2$~digits. 
To illustrate, let us consider the product between two integers having three digits. In base 10, the smallest product is 100 times 100 or \num{10000}, so it requires 5~digits. The largest product is 999 times 999 or \num{998001} (6~digits).

Suppose that we want $d\geq 1$~digits of accuracy (in base $b$) for the integer product $w\times z$ given a fixed integer $z$. As much as possible, we want to compute the $d$~most significant digits:
\begin{itemize}
    \item If $b^{d-1}\leq  (w\times z) \leq  b^d-1$, then we output $w\times z$. 
    \item If $ b^d\leq  (w\times z) \leq b^{d+1}-1$, then we output $(w\times z) \div b$.
    \item \ldots
    \item Generally, if $b^{d+k-1}\leq  (w\times z) \leq b^{d+k}-1$ for $k\geq 0$, then we output $(w\times z) \div b^k$.
    
\end{itemize}

\section{Short Multipliers}
\label{sec:shortmultipliers}
Suppose that we want to compute the most significant digits
of $w \pi$ for small integers $w$.
Materializing all digits of $\pi$ is impossible, so we use
a truncated version: $ 3.1416 \times w$.
Thus we may compute the most-significant 
digits of $  3.1416 \times w$. For simplicity, we can omit
the decimal point.
How large could the integer $w$ be if we want two digits of accuracy
assuming we do not use the missing digits of $\pi$?
The answer is that $w$ should not exceed 2068.
It is an instance of the general question we want to be able to answer.

When $10^3 \leq 31416 \times w \leq 10^4 -1$, 
 $ (31416 \times w)\div 10$ provides the two-most significant digits of the product.
We want to determine when  $ (31416 \times w)\div 10$  matches
the value we would get when computing the full product:  $ (31416 \times w)\div 10 = (10000 \pi \times w)\div 10 $
irrespective of the truncated portion of $\pi$ ($10000 \pi-31416$).
The following lemma provides a necessary and sufficient condition.

\begin{lemma}\label{lemma:exact}
Let $z$ be a truncated integer multiplier and $z'=z + \epsilon$ be the exact multiplier with $\epsilon \in [0,1)$. For integers $w$ and $M$,
we have that $( w \times z) \div M = ( w \times z') \div M$ for all $\epsilon$
if and only if $(w\times z) \rem M < M - w + 1$.
\end{lemma}

\begin{proof}
Let $z$ be the truncated integer multiplier and $z'=z + \epsilon$ be the exact multiplier with $\epsilon \in [0,1)$.
Since $ z \leq z' < z+1$, we have that
$ w \times z \leq w \times z' <  w \times z + w$.
Hence we have that 
$ w \times z \leq w \times z' \leq   w \times z + w -1$.
Thus we have
$( w \times z) \div M \leq  ( w \times z') \div M \leq (w \times z + w - 1) \div M$. 
Therefore, we have that $( w \times z) \div M = ( w \times z') \div M$ for all $z'$
if and only if $( w \times z) \div M = (w \times z + w - 1) \div M$.
Write $(w \times z + w - 1) \div M = ((w \times z)\div M \times M + (w\times z) \rem M + w - 1) \div M $. 
With a little arithmetic, we see from this last equality that 
 $(w\times z) \rem M + w - 1<M$
if and only if $(w \times z + w - 1) \div M = ( w \times z) \div M $.
We have proven the lemma.
\end{proof}

We can apply this condition to the computation of digits. Consider $w\times z$
and suppose you desire to have  $d\geq 1$~digits of accuracy in base $b$.
\begin{itemize}
    \item If $(w\times z)  \leq  b^{d-1}-1$, then we cannot produce $d$~digits by truncation. Thus we may require as a pre-condition that 
    $(w\times z) \geq b^{d-1}$ or $w\geq (b^{d-1}+z-1)\div z$.
    \item If $ b^{d-1}\leq (w\times z) \leq  b^d-1$, then we output $w\times z$ after checking that  $w<2$.
    \item If $ b^d\leq  (w\times z) \leq b^{d+1}-1$, then we output $(w\times z) \div b$ after checking that  $(w\times z) \rem b < b - w + 1$.
    \item \ldots
    \item Generally, if $b^{d+k-1}\leq  (w\times z) \leq b^{d+k}-1$ for $k \geq 0$, then we output $(w\times z) \div b^k$ after checking that  $(w\times z) \rem b^k  < b^k - w + 1$.
\end{itemize}

We would like to check  $(w\times z) \rem b^k  < b^k - w + 1$ efficiently over the range $b^{d+k-1}-1< (w\times z) \leq b^{d+k}-1$ to verify whether we can compute digits exactly.

\subsection{Finding the Valid Range}
\label{sec:validrange}
Suppose that  $(\beta\times z)\rem M$ is the maximum of 
 $(w\times z)\rem M$  for $w=0,1, \ldots, \beta$, then it follows
 that if $(\beta\times z)\rem M\leq M- \beta +1$ is satisfied,
 it must be that 
 $(w\times z)\rem M< M - w +1$ is satisfied for $w=0,1, \ldots, \beta -1$, since $M-w+1$ is strictly decreasing and  $(w\times z)\rem M<(\beta\times z)\rem M$.
 Conversely, the next lemma shows that the first time that  $(w\times z)\rem M< M - w +1$ is falsified,  $(w\times z)\rem M$ is a new maximum. 
      
      \begin{lemma}\label{lemma:helpful}Let $w'$ be the smallest integer value $w'\geq 0 $ such that  $(w'\times z)\rem M\geq  M - w' +1$, then we have that
       $(w'\times z)\rem M> (w\times z)\rem M$ for $w=0,1,\ldots, w'-1$.
      \end{lemma}
\begin{proof}
Let $w_2$ be the smallest value $w$ such that
 $(w\times z)\rem M< M - w +1$ is falsified. Let $w_1$ be the location of the maximum of  $(w\times z)\rem M$ up to $w_2$ exclusively: $0\leq w_1< w_2$.
 Then we have $w_1\leq  w_2$ such that
 $(w_1\times z)\rem M \geq (w_2\times z)\rem M$,
  $(w_1\times z)\rem M< M - w_1 +1$
  and 
   $(w_2\times z)\rem M\geq M - w_2 +1$.
   Thus we have that 
   $(w_1\times z)\rem M - (w_2\times z)\rem M<w_2-w_1$ or
   $(w_2\times z)\rem M- (w_1\times z)\rem M  > w_1-w_2$
   or 
      $M+(w_2\times z)\rem M- (w_1\times z)\rem M  > M+w_1-w_2$.
      Thus we have that 
      $((w_2-w_1)\times z)\rem M  > M-(w_2-w_1)$. This indicates
      that $w_2-w_1$ falsifies $(w\times z)\rem M< M - w +1$, which is only possible if $w_1=0$ by our assumption, but that is not possible since it would imply that the maximum is 0. We have shown the  lemma.
\end{proof}
This lemma is helpful because it indicates that we only need to check the condition $(w\times z)\rem M\geq  M -w +1$ when $(w\times z)\rem M$ is a new maximum.

Given a short multiplier $z$ and a desired number of digits $d$ in base $b$, we may seek the upper range of the variable $w$ such that the $d$~most significant digits of $w\times z$ are exact. We iterate over $k = 0, 1, \ldots$ and for values of $w$ such that $b^{d+k-1}-1< (w\times z) \leq b^{d+k}-1$, we seek the smallest value $w$ such that  $(w\times z) \rem b^k  \geq b^k - w + 1$. When such a value exists, the algorithm terminates with a value $w$ indicating the upper bound.
The lower bound is given by $w\geq (b^{d-1}+z-1)\div z$.
The lower and upper bounds define the range of validity: given any integer $w$ in this range, the $d$~most significant digits of $w \times z$ are exact.
%






\section{Enumerating the Extrema of Remainders}
\label{sec:enumerating}

As we compute $(w \times z) \rem M$ for $w=1, 2, 3, \ldots, M-1$, we encounter new minima and new maxima. We seek to efficiently locate all such extrema.

Consider $M=8$ and $z=3$. Given the first two values, $(1 \times 3) \rem 8= 3$ and $(2 \times 3) \rem 8= 6$, we have that
the former is a minimum while the later is a maximum. See Table~\ref{tab:example}. The next value at $w=3$ is 1, a new minimum. The next extrema is at $w=5$, a maximum. Observe that prior to $w=5$, we had a maximum at $w=2$ and a minima at $w=3$ and that $3+2=5$. As we shall show, all extrema follow this rule: they appear at a location that is the sum of the location of last minimum with the last maximum.

\begin{table}
    \centering
    \begin{tabular}{ccc}\toprule
        $w$  &  $(w \times 3) \rem 8$ & classification\\\midrule
1  & 3  & maximum/minimum\\
2  & 6  & maximum\\
3  & 1  & minimum \\
4  & 4  & \\
5  & 7  & maximum\\
6  & 2 & \\
7 & 5& \\ 
8 & 0 & minimum \\ 

\bottomrule
    \end{tabular}
    \caption{Example of remainders with $M=8$ and $z=3$.}
    \label{tab:example}
\end{table}

When $z$ and $M$ are coprime, then the sequence of values $(w \times z) \rem M$ for $w=1, 2, 3, \ldots, M-1$ are a permutation of
the integers from 1 to $M-1$. When $z$ and $M$ have a non-trivial common divisor, the values repeat. Indeed, whenever $(w \times z) \rem M = (w' \times z) \rem M$, we have that $((w-w') \times z)\rem M = 0$. Assume without loss of generality that $w>w'$, then we have that $(w-w') \times z$ is a positive integer divisible by $M$. Thus we have that the sequence $(w \times z) \rem M$ for $w\in [0,M/\gcd(M,z))$ is made of distinct values. This sequence of values repeats over the next interval $ [M/\gcd(M,z), 2M/\gcd(M,z))$ and so forth, $\gcd(M,z)$ times until $w=M$.

The next lemma shows how we can always determine the location of the next extrema from the last minimum and the last maximum. If the last minimum is at $w=\alpha$ and the last maximum is at $w=\beta$, then the next extrema is at $w=\alpha +\beta$. See Fig.~\ref{fig:maxmin}.
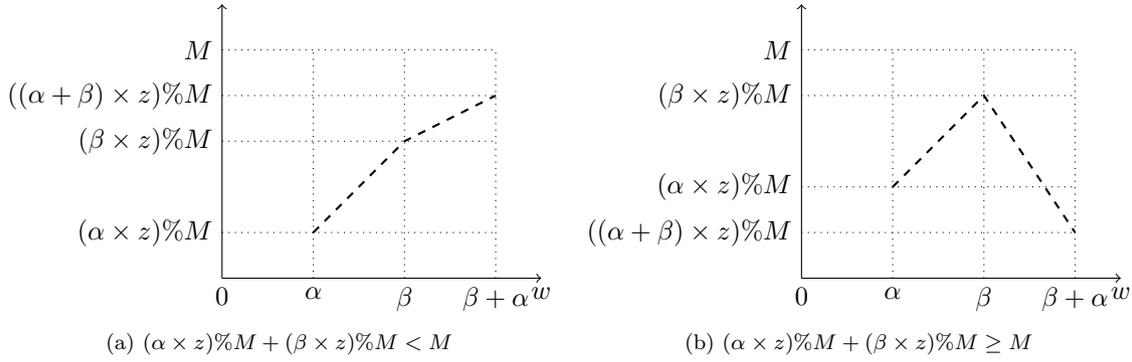
\begin{figure*}\centering
\subfloat[$(\alpha\times z)\% M + (\beta\times z)\% M <M$]{
\begin{tikzpicture}[scale=0.6]
\draw[->] (0,0) -- (7,0) node[anchor=north] {$w$};
\draw	(0,1) node[anchor=east] {$(\alpha\times z)\% M$}
		(0,3) node[anchor=east] {$(\beta\times z)\% M$}
		(0,4) node[anchor=east] {$((\alpha+\beta)\times z)\% M$}
		(0,5) node[anchor=east] {$M$};
\draw	(0,0) node[anchor=north] {0}
		(2,0) node[anchor=north] {$\alpha$}
		(4,0) node[anchor=north] {$\beta$}
		(6,0) node[anchor=north] {$\beta+\alpha$};
\draw[->] (0,0) -- (0,6) node[anchor=east] {};
\draw[dotted] (2,0) -- (2,5);
\draw[dotted] (4,0) -- (4,5);
\draw[dotted] (6,0) -- (6,5);
\draw[dotted] (0,1) -- (6,1);
\draw[dotted] (0,3) -- (6,3);
\draw[dotted] (0,4) -- (6,4);
\draw[dotted] (0,5) -- (6,5);
\draw[thick,dashed] (2,1) -- (4,3) -- (6,4);
\end{tikzpicture}
}
\subfloat[$(\alpha\times z)\% M + (\beta\times z)\% M \geq M$]{
\begin{tikzpicture}[scale=0.6]
\draw[->] (0,0) -- (7,0) node[anchor=north] {$w$};
\draw	(0,2) node[anchor=east] {$(\alpha\times z)\% M$}
		(0,4) node[anchor=east] {$(\beta\times z)\% M$}
		(0,1) node[anchor=east] {$((\alpha+\beta)\times z)\% M$}
		(0,5) node[anchor=east] {$M$};
\draw	(0,0) node[anchor=north] {0}
		(2,0) node[anchor=north] {$\alpha$}
		(4,0) node[anchor=north] {$\beta$}
		(6,0) node[anchor=north] {$\beta+\alpha$};
\draw[->] (0,0) -- (0,6) node[anchor=east] {};
\draw[dotted] (2,0) -- (2,5);
\draw[dotted] (4,0) -- (4,5);
\draw[dotted] (6,0) -- (6,5);
\draw[dotted] (0,2) -- (6,2);
\draw[dotted] (0,4) -- (6,4);
\draw[dotted] (0,1) -- (6,1);
\draw[dotted] (0,5) -- (6,5);
\draw[thick,dashed] (2,2) -- (4,4) -- (6,1);
\end{tikzpicture}
}
\caption{\label{fig:maxmin}Illustration of how the next extrema is computed from the last minimum and the last maximum}
\end{figure*}

\begin{lemma}
\label{lemma:technical}
Suppose that, over a range $w=1,2, \ldots, \max(\beta,\alpha)$
for $\max(\beta,\alpha)<M/\gcd(M,z)$, we have that
$(\beta \times z) \rem M$ is the maximal value of $(w\times z) \rem M$, while 
$(\alpha \times z) \rem M$ is the minimal value, then if we extend the sequence to $w=1,2, \ldots, \alpha+\beta$, we have that   
\begin{itemize}
    \item When $((\alpha+\beta) \times z) \rem M > (\beta \times z) \rem M$,  $((\alpha+\beta) \times z) \rem M$ is the new maximum while $(\alpha \times z) \rem M$ remains the minimum,
    \item otherwise we have that $((\alpha+\beta) \times z) \rem M < (\beta \times z) \rem M$, and $((\alpha+\beta) \times z) \rem M$ is the new minimum while $(\beta \times z) \rem M$ remains the maximum.
\end{itemize}

\end{lemma}
\begin{proof}
For brevity, assume that $\alpha<\beta$: the counterpart ($\alpha>\beta$) follows by symmetrical arguments.

We want to show that $((\alpha+\beta) \times z)$ is either smaller
than $ (\alpha \times z) \rem M$ or larger than  $(\beta \times z) \rem M$:
\begin{itemize}
\item Suppose that  $((\alpha+\beta) \times z) \rem M \geq (\alpha \times z) \rem M$ then $((\alpha+\beta) \times z) \rem M = (\alpha \times z) \rem M + (\beta \times z) \rem M >(\beta \times z) \rem M  $.
\item Suppose that  $((\alpha+\beta) \times z) \rem M \leq (\beta \times z) \rem M$ then $((\alpha+\beta) \times z) \rem M = (\alpha \times z) \rem M +(\beta \times z) \rem M -M< (\alpha \times z) $.
\end{itemize}




First suppose that $((\alpha+\beta) \times z) \rem M > (\beta \times z) \rem M$. We want to show that $((\alpha+\beta) \times z) \rem M$ is the new maximum over the extended range while $(\alpha \times z) \rem M$ remains the minimum.
Suppose that it is not the case, then at least one of the following two cases hold:
\begin{itemize}
    \item If $((\alpha+\beta) \times z) \rem M $ is not the new maximum, there must be a new, even larger maximum.
    Thus there must be a value $\omega$ satisfying $0<\omega < \alpha$ such that 
     $((\omega+\beta) \times z) \rem M >  ((\alpha+\beta) \times z) \rem M >(\beta \times z) \rem M $.
     It implies that
     $((\omega+\beta) \times z) \rem M =
     ((\omega-\alpha) \times z) \rem M+
     ((\alpha+\beta) \times z) \rem M 
     = ((\omega-\alpha) \times z) \rem M+
     (\alpha \times z) \rem M+
     (\beta \times z) \rem M =
     ((\omega+\beta-\alpha) \times z) \rem M+
     (\alpha \times z) \rem M
     $. Hence we have that
     $ ((\omega+\beta-\alpha) \times z) \rem M = ((\omega+\beta) \times z) \rem M-  (\alpha \times z) \rem M>
    ((\alpha+\beta) \times z) \rem M-  (\alpha \times z) \rem M=(\beta \times z)\rem M $. Hence we have that
         $((\omega+\beta-\alpha) \times z) \rem M >  (\beta \times z) \rem M $ which contradicts that $\beta$ is a maximum for the range up to $\beta$ since $\omega+\beta-\alpha<\beta$. Thus this case is not possible.
    \item If  $(\alpha \times z) \rem M$ did not remain the minimum, then there must be a new, even smaller minimum: there must be a value $\omega$ satisfying $0<\omega < \alpha$ such that 
     $((\omega+\beta) \times z) \rem M <  (\alpha\times z)\rem M<(\beta\times z)\rem M$.
     It implies that
     $ (\omega \times z) \rem M 
     +(\beta\times z)\rem M - M
     <  (\alpha\times z)\rem M$
     or
     $ (\alpha\times z)\rem M - (\omega \times z) \rem M> (\beta\times z)\rem M - M$.
Because  $0<\omega < \alpha$, we must have that $ (\omega \times z) \rem M >(\alpha \times z) \rem M$.
Thus we have that 
$((\alpha - \omega) \times z) \rem M$
is 
$M+(\alpha \times z) \rem M - (\omega \times z) \rem M >(\beta\times z)\rem M$
      which again contradicts the fact that  $\beta$ is a maximum for the range up to $\beta$ since $\alpha-\omega<\alpha<\beta$.
\end{itemize}
Hence the result holds.

Second suppose that $((\alpha+\beta) \times z) \rem M < (\beta \times z) \rem M$. We want to show that  $((\alpha+\beta) \times z) \rem M$ is the new minimum while $(\beta \times z) \rem M$ remains the maximum.
Suppose that it is not the case, then at least one of the following two cases hold:
\begin{itemize}
\item If $((\alpha+\beta) \times z) \rem M$  is not the new minimum, then there must be a new even smaller minimum.  
Thus there must be a value $\omega$ satisfying $0<\omega < \alpha$ such that 
     $((\omega+\beta) \times z) \rem M <  ((\alpha+\beta) \times z) \rem M $. But this inequality implies that
      $(\omega\times z) \rem M < (\alpha \times z) \rem M $ which would contradict the fact that $(\alpha \times z) \rem M $ was the minimum.
\item If $(\beta \times z) \rem M$ is no longer the maximum, then there must be a new larger maximum. Thus there must be a value $\omega$ satisfying $0<\omega < \alpha$ such that  $((\omega+\beta) \times z) \rem M > (\beta \times z) \rem M $.
Hence $((\omega+\beta) \times z) \rem M= (\omega \times z) \rem M+ (\beta \times z) \rem M $.
Because $(\alpha\times z) \rem M$ is the minimum up to $\beta$, we must have  that $(\omega \times z) \rem M>(\alpha\times z) \rem M$.
Therefore we have that
 $((\omega+\beta) \times z) \rem M= (\omega \times z) \rem M+ (\beta \times z) \rem M >  (\alpha \times z) \rem M + (\beta \times z) \rem M$. However, since  $((\alpha+\beta) \times z) \rem M < (\beta \times z) \rem M$, we must have that $(\alpha \times z) \rem M + (\beta \times z) \rem M \geq M$ and so $((\omega+\beta) \times z) \rem M\geq M$, a contradiction.
\end{itemize}
Hence the result holds. 
\end{proof}

Lemma~\ref{lemma:technical} implies that you can visit all of the extrema of $(w \times z) \rem M $ for $w=1, 2, \ldots$ by first finding the first two extrema (a maximum at $\beta$ and a minimum at $\alpha$), and then locate a new maximum or a new minimum at $\alpha+\beta$, and so forth.
Through an iterative process, you are guaranteed to only ever visit running extrema.

Unfortunately, such a process could be slow. Consider the case when $z=1$. We have that the sequence $(w \times z) \rem M $ for $w=1, 2, \ldots$ is $1,2, \ldots$ It implies that every single possible value of $w$ is, when it is encountered, a new maximum. When applying Lemma~\ref{lemma:technical} to this case, we find that $\alpha =1 $ (throughout), with $\beta$ taking the values $2,3, \ldots$ Such an algorithm would encounter $M$~extrema and would run in time $\Omega(M)$. Thankfully, we can characterize the location of the extrema in logarithmic time, as we show next.

We have that $z \rem M =( 2z )\rem M$ if an only if $M$ divides $z$. Assume that $M$ does not divide $z$.
As long as $w < M/\gcd(z,M)$, we have that  $(w\times z) \rem M \neq 0$.

Thus, after two elements in the series $(w \times z) \rem M$ for $w=1, 2, \ldots$, we have a minimum and a maximum value. We write the maximum value $b$, and its corresponding $w$ is $\beta$. 
We write the minimum value $a$, and its corresponding $w$ is $\alpha$. 
We have that $b>a>0$.
As we keep progressing over  $w\in [3,M/\gcd(M,z))$, we may encounter a new maximum or a new minimum.

Assume that the first value was a minimum (i.e., $a= z \rem M$ and $\alpha =1$) followed by a maximum  (i.e., $b= 2z \rem M$ and $\beta =2$).
If $3a<M$, then we have a new maximum immediately after at $\beta = 3$. Similarly if $4a<M$ and so forth.
We have exactly
$(M-1-a)\div a$ consecutive maxima: $b=2a$ at $\beta =2$, $3a$ at $\beta =3$,\ldots,  $a+( (M-1-a)\div a) \times a$ at $\beta =1+(M-1-a)\div a$. 

An analogous scenario unfolds when we assume that the first value was a maximum ($b= z \rem M$ and $\beta =1$) followed by a minimum ($a= 2z \rem M$ and $\alpha =2$).
If $a+b\geq M$, then we have $(a+b) \rem M = a+b-M< a+b$, and thus we have a new minimum (smaller by $M-b$). 
We have
$b\div (M-b)$ consecutive minima: $b+b-M$ at $\alpha =2$, $b+2(M-b)$ at $\alpha =3$,\ldots,  $b+( b\div (M-b)) \times (b-M)$ at $\beta =1+b\div (M-b)$.
 We have that the last maximum is greater than $M/2$.

Using Lemma~\ref{lemma:technical} and the first two extrema, we can efficiently iterate through all other extrema: 
\begin{itemize}
\item Suppose that we found our last minimum at $\alpha$. We find a new maximum at $\beta$ ($\alpha < \beta$). By Lemma~\ref{lemma:technical}, this maximum is followed by up to
$ (M-b-1) \div a$ even greater maxima:
$b+a$ at $w=\beta+\alpha$, \ldots,  $b+((M-b-1) \div a) \times a$ at $w=\beta + ((M-b-1) \div a)\times \alpha$.
The maxima are each time incremented by $a$, and they appear at locations incremented by $ \alpha$.
If we were to continue one more step (increment by $a$ once more), we would exceed $M-1$.
If we redefine $b\leftarrow b+((M-b-1) \div a) \times a$ and $\beta \leftarrow \beta + ((M-b-1) \div a)\times \alpha$, then we have that the value at $\alpha + \beta$ is $(a+b) \rem M = a+b-M<a$, thus a new minimum.
\item Suppose that we found our last maximum at $\beta$. We find a new minimum at $\alpha$ ($\beta < \alpha$). 
 By Lemma~\ref{lemma:technical}, this minima is followed by 
$a\div (M-b)$ even smaller minima:
$a+b-M$ at $w=\alpha+\beta$, \ldots,  $a+(a\div (M-b)) \times (b-M)$ at $w=\alpha + (a\div (M-b))\times \beta$.
They happen at locations separated by $\beta$ and decremented by $b-M$. If we were to continue one more step, we would increment by $b$ (as opposed to $b-M$), and we would not have a new minimum.
If we redefine $a\leftarrow a+(a\div (M-b)) \times (b-M)$ and $\alpha \leftarrow \alpha + (a\div (M-b))\times \beta$, then we have that the value at $\alpha + \beta$ is $(a+b) \rem M = a+b>b$, thus a new maximum.
\end{itemize}

Thus we have that $(w \times z)\rem M$ for $w=1,2,\ldots,  M/\gcd(z,M) - 1$ alternates between new equispaced sequences of maxima and new equispaced sequences of minima.
We do not need to compute $\gcd(z,M)$ explicitly: we
know that when $w=M/\gcd(z,M)$, we have that
$(w\times z)\rem M= 0$. Thus it suffices to check for
a minimum value of 0.
The algorithm given in  Fig.~\ref{fig:gaps}
outputs the sequence of \emph{gaps} (successive $\alpha$ and $\beta$) which determine the locations of the extrema. Each gap value can generate several equispaced extrema.

We can check that with every iteration, going through a sequence of maxima, a sequence of minima and then back to a sequence of maxima, the gap $\beta$ has more than doubled. And similarly for $\alpha$. Thus we have that the algorithm given in  Fig.~\ref{fig:gaps} runs in time $O(\log (M/\gcd(z,M)))$ when assuming that arithmetic operations run in  constant time.

Consider the algorithm of Fig.~\ref{fig:gaps} with an example: 
 $M=8$ and $z=3$. 
 At first we have $a= b= (1\times 3)\rem 8 = 3$ and $w=\alpha=\beta=1$. The list  \texttt{lbda} is initially empty. 
\begin{enumerate}
    \item We consider $v=(a+b)\rem M = (3+3)\rem 8  = 6$. It is a new maximum ($v>b$). 
    We compute $t=(M-b-1)\div a = (8-3-1)\div 3$ which is one.
    We append $w=1$ to \texttt{lbda}.  We move to $w=1+1=2$ and set $\beta=2$. We have $b=6$ and $a=3$.

    \item We consider $v=(a+b)\rem M = (3+6)\rem 8=1$. It is a new minimum ($v<a$). We append $w=2$ to $\lambda$. We compute  $ a \div (M-b) = 3 \div 2=1$. We move to $w=2+1=3$ and we set $\alpha=3$. We have $a=1$.
    \item We consider  $v=(a+b)\rem M = (1+6)\rem 8=7$. We have a new maximum. We append $w=3$ to \texttt{lbda}. We compute $(M-b-1)\div a = (8-6-1)\div 1$ which is one, again. We move to $w=w+\beta = 3+2=5$. We set $\beta=5$ and $b=7$.
    \item We consider $(a+b)\rem M= (1+7)\rem 8 = 0$. It is a new minimum, we append $\beta=5$ to \texttt{lbda} and we exit the main  loop, returning $\mathtt{lbda}=\{1,2,3,5\}$.
\end{enumerate}

If we compute $\distance_8(0,(w*3)\rem 8)$ for $w=1,2,3,5$ we get 3,2,1,1. That is, the distance of the various extrema to zero diminishes. It is clear that it must be so for successive (smaller) minima and also for successive (larger) maxima: the distance with zero must be strictly decreasing. Indeed a maximum is a value close to $M$, the closer to $M$, the larger it is. When we reach $M-1$, the largest value, we have that $\distance_M(0,M-1)=1$, the minimal distance.
But it is also true of a minimum followed by a maximum, or a maximum by a minimum: the distance with zero must remain the same or decrease. E.g., it follows by inspection: if we have a minimum value $a$, then it is not possible for the largest maximum of the next sequence of maxima to be more than $a$ away from $M$.

\begin{figure}
\begin{python}
def gaps(z, M):  
  w = 1
  lbda = []
  a = z % M
  alpha = 1
  b = z % M
  beta = 1
  while True:
    v = (a + b) % M
    if v < a:
      lbda.append(w)
      if a % (M - b) == 0: break
      t = a // (M - b)
      w = w + alpha + (t - 1) * beta
      alpha = w
      a = (a + t * b) % M
    else:
      t = (M - b - 1) // a
      lbda.append(w)
      w = w + beta + (t - 1) * alpha
      beta = w
      b = (b + t*a) % M
  return lbda
\end{python}
\caption{\label{fig:gaps}Python code to 
compute all of the gaps between the extrema of $(w \times z)\rem M$ for $w=1,2,\ldots, \frac{M}{\gcd(z,M)}  - 1$. $M$ and $z$ should be  positive integers, and $z$ should not be a multiple of $M$.}
\end{figure}

\subsection{Bounding Remainders with an Offset}
\label{sec:offset}
The algorithm of Fig.~\ref{fig:gaps} provides an efficient algorithm to enumerate all the extrema of remainders  $ (w \times z) \rem M$ for $w=1,2,\ldots$ 
    We might want to enumerate the extrema starting from an arbitrary point:  $(w \times z) \rem M$ for $w=A, A+1, \ldots, B$ in which case
    we can rewrite the problem as $(A\times z + w \times z) \rem M$ for $w=0, 1, \ldots, B-A$.
Setting $b=A\times z$, we find that it is equivalent to
finding the extrema of  $(b + w \times z) \rem M$ (for $w=0, 1, \ldots$).
Thus we want to extend our previous results to remainder of a product with an offset ($b$).

We cannot rely directly on the earlier result for  $( w \times z) \rem M$ (Lemma~\ref{lemma:technical}) which indicates that the next extrema is effectively the sum of the previous minima and the previous maxima. Indeed, consider  $(7 + w \times 3) \rem 8$: we have the value 2 at $w=1$, followed by value 5 (a new maximum) at $w=2$, value 0 (a new minimum) at $w=3$, value 6 (a new maximum) at $w=5$, value 7 (a new maximum) at $w=8$. However, we can still make good use of Lemma~\ref{lemma:technical}.

Suppose that 
$(b + \beta \times z) \rem M$ is the maximum so far over $(b + w \times z) \rem M$ for $w=1,2,\ldots,\beta$. Suppose
that the next extrema is at
$(b + (\beta+k) \times z) \rem M$.
Then we must have that 
$(k \times z) \rem M$
is a minimum of $(w \times z) \rem M$ over $w=1,2,\ldots, k$ otherwise we would have a new intermediate extrema. And similarly when we start from a minima. 
It follows that we can access the extrema of $(b + w \times z) \rem M$ 
by considering offsets by the gap values generated by the algorithm of Fig.~\ref{fig:gaps}. Because the gaps are monotonic, and because our  maxima are only larger, and our minima only smaller, there is no need to consider previous gaps once they can no longer increase a maximum or decrease a minimum.

Even with a non-zero offset, the values still repeat: $ (b+ w \times z) \rem M = (b+ (w + M/\gcd(z,M))\times z) \rem M$. It is not necessary to compute $\gcd(z,M)$, we could
instead stop when $w> M$ since no new extrema can be found after $w$ reaches $M/\gcd(z,M)$ given that the values repeat.

Thus the distance between the extrema of
 $(w \times z) \rem M$ for $w=A, A+1, \ldots, B$ 
must be within the values produced by the \texttt{gaps} function of Fig.~\ref{fig:gaps}, and that the gaps only grow larger. 
As an application, 
the \texttt{find\_min\_max} function from Fig.~\ref{fig:extrema}  computes the locations of all of the maxima and minima 
 of $(w\times z) \rem b^k$ from $w=A$ to $w=B$ (inclusively). 
 It encodes sequences of equispaced minima (or equispaced maxima) as a triple with the location of the last extrema, 
 their number and the size of the gaps between the extrema. It runs in time $O(\log (M/\gcd(z,M)))$, assuming that arithmetic operations run in constant time.

\begin{figure}
\begin{python}
def find_min_max(z, M, A, B):
  mnma = [(A, 0, 0)]
  mxma = [(A, 0, 0)]
  facts = gaps(z, M)
  mi = (z * A) % M
  ma = (z * A) % M
  fact_index = 0
  b = A * z
  w = 0
  while True:
    offindex = facts[fact_index]
    v = (z * (w + offindex) + b) % M
    if w + offindex > B-A: break
    if v < mi:
      w += offindex
      mi = v
      basis = (z * w + b) % M
      off = (z * offindex) % M
      times = basis // (M - off)
      if A + w + times * offindex > B:
        times = (B - A - w) // offindex
      w += offindex * times
      mi = (z * w + b) % M
      mnma.append((w+A, times, offindex))
    elif v > ma:
      w += offindex
      ma = v
      basis = (z * w + b) % M
      off = (z * offindex) % M
      times = (M - 1 - basis) // off
      if A + w + times * offindex > B:
        times = (B - A - w) // offindex
      w += offindex * times
      ma = (z * w + b) % M
      mxma.append((w + A, times, offindex))
    else:
      fact_index += 1
      if fact_index == len(facts): break
  return (mnma, mxma)
\end{python}
\caption{\label{fig:extrema}Python code to
enumerate all extrema  of $(w\times z) \rem M$ for $w = A,\ldots, B$. $M$ and $z$ should be  positive integers, and $z$ should not be a multiple of $M$.}
\end{figure}

\section{Computing the Range}
\label{sec:computingtherange}
 Our main function (\texttt{find\_range}) is provided Fig.~\ref{fig:findrange}: 
given a multiplier $z$ and a desired number digits in a given base, it computes a range of values $[\mathrm{lb},\mathrm{ub})$ such that
if $w \in [\mathrm{lb},\mathrm{ub})$, then $w\times z$ has its most significant digits exact even when $z$ is a truncated multiplier---irrespective of the unknown digits. When the interval is empty, 
the \texttt{None} value is returned.

Following \S~\ref{sec:shortmultipliers}, the function checks $(w\times z) \rem b^k  < b^k - w + 1$ over 
the range $b^{d+k-1}\leq (w\times z) \leq b^{d+k}-1$ to verify whether we can compute digits exactly: starting 
with $k=0$, we increment $k$ until a value $w$ satisfying $(w\times z) \rem b^k  \geq  b^k - w + 1$ is found.
 To do so, the \texttt{find\_range} function relies on 
 the \texttt{find\_min\_max} function from Fig.~\ref{fig:extrema}. 
We handle separately the case when $z$ is a multiple of $b^k$, in which case $(w\times z) \rem b^k=0$
and  $(w\times z) \rem b^k  < b^k - w + 1$ becomes 
$w < b^k + 1$: thus we must stop if $B\geq b^k + 1$.

\begin{figure}
\begin{python}
def find_range(z, digits, base):
 lb = (base ** (digits - 1) + z - 1) // z
 k = 0
 while True:
  A = (base ** (digits + k - 1) + z - 1)//z
  B = base ** (digits + k) // z
  if B * z == base ** (digits + k):
   B -= 1
  if B < A:
   k = k + 1
   continue
  M = base ** k
  if z % M == 0: # special case
    if B >= M+1:
      return (lb, M+1)
  (mnma, mxma) = find_min_max(z, M, A, B)
  for (beta, times, gap) in mxma:
   if beta * z % M >= M - beta + 1:
    ub = beta
    if times > 0:
     top = beta * z % M - (M - beta + 1)
     bottom = gap + gap * z % M
     mt = top // bottom
     ub = beta - top // bottom * gap
    if ub <= lb:
     return None
    return (lb, ub)
  k = k + 1
\end{python}
\caption{\label{fig:findrange}Python code to compute the range of values of $w$ for which the most significant digits of $w\times z$ are exact even if $z$ is truncated. The function allows the user to specify the number of exact most-significant
 digits desired (\texttt{digits}) as well as the basis (\texttt{base}).}
\end{figure}
We can generate Table~\ref{table:acc} with a Python script:
\begin{python}
for mypi in [
  3141592653,
  31415926535,
  314159265358,
  3141592653589,
  31415926535897,
  314159265358979,
  3141592653589793,
  31415926535897932,
  314159265358979323,
  3141592653589793238,
  31415926535897932384,
  ]:
  print(find_range(mypi, 10, 10))
\end{python}

\section{Conclusion and Future Work}

It is intuitive that we are able to compute
the most significant digits of a product using
only the most significant digits of one of the
multipliers (i.e., a short multiplier). 
Thankfully, we can check for the
exact range of validity of a short multiplier
using efficient logarithmic-time algorithms. 

We have identified some future work:
\begin{itemize}

\item We construct short multipliers by truncating existing multipliers. However it may be possible to round the multiplier instead of truncating it.
Similarly, we compute the most significant digits of a product by truncation: we may round instead. There are many rounding rules that should be considered when the result is ambiguous: round up, round down, round to even~\cite{reiser1975evading}, round to odd.

\item From a short multiplier and a desired number of most-significant digits, we have derived a range of validity. We can also start by a desired number of most-significant digits and a range of validity, and derive the smallest short multiplier: it suffices to construct ever more precise short multipliers. A direct algorithm to efficiently derive the best short multipliers might be useful.
\end{itemize}

\section{Funding}
This research was funded by the Natural Sciences and Engineering Research Council of Canada, Grant Number: RGPIN-2017-03910.

\section{Data Availability}

No new data were generated or analysed in support of this research.

\bibliographystyle{compj}
\bibliography{exactshort}

\begin{thebibliography}{99}

\bibitem{hars2006applications}
Hars, L. (2006) Applications of fast truncated multiplication in cryptography.
\newblock {\em EURASIP Journal on Embedded Systems}, {\bf  2007}, 061721.

\bibitem{schulte1993truncated}
Schulte, M.~J. and Swartzlander, E.~E. (1993) {Truncated multiplication with
  correction constant [for DSP]}.
\newblock {\em Proceedings of IEEE Workshop on VLSI Signal Processing},  pp.
  388--396. IEEE.

\bibitem{10.1145/3192366.3192369}
Adams, U. (2018) Ryu: Fast float-to-string conversion.
\newblock {\em Proceedings of the 39th ACM SIGPLAN Conference on Programming
  Language Design and Implementation},  New York, NY, USA PLDI 2018  270–282.
  Association for Computing Machinery.

\bibitem{10.1145/3360595}
Adams, U. (2019) Ryu revisited: Printf floating point conversion.
\newblock {\em Proc. ACM Program. Lang.}, {\bf  3}.

\bibitem{lemire2021number}
Lemire, D. (2021) Number parsing at a gigabyte per second.
\newblock {\em Software: Practice and Experience}, {\bf  51}, 1700--1727.

\bibitem{reiser1975evading}
Reiser, J.~F. and Knuth, D.~E. (1975) Evading the drift in floating-point
  addition.
\newblock {\em Information Processing Letters}, {\bf  3}, 84--87.

\end{thebibliography}

\end{document}